\documentclass[10pt,conference]{IEEEtran}
\usepackage{bm}
\usepackage{amssymb}
\usepackage{amsmath}

\newtheorem{lemma}{Lemma}
\newtheorem{theorem}{Theorem}
\newtheorem{corollary}{Corollary}
\newtheorem{claim}{Claim}

\begin{document}

\title{Decoding Frequency Permutation Arrays  under Infinite norm}

\author{
\authorblockN{Min-Zheng Shieh}
\authorblockA{Department of Computer Science\\
National Chiao Tung University \\
1001 University Road, Hsinchu, Taiwan\\
Email: mzhsieh@csie.nctu.edu.tw}
\and
\authorblockN{Shi-Chun Tsai}
\authorblockA{Department of Computer Science\\
National Chiao Tung University \\
1001 University Road, Hsinchu, Taiwan\\
Email: sctsai@cs.nctu.edu.tw}
}
%

\maketitle

\begin{abstract}
A frequency permutation array (FPA) of length $n=m\lambda$ and distance $d$ is a set of permutations on a multiset over $m$ symbols, where each symbol appears exactly $\lambda$ times and the distance between any two elements in the array is at least $d$. FPA  generalizes the notion of permutation array. In this paper, under the distance metric $\ell_\infty$-norm, we first prove lower and upper bounds on the size of FPA. Then we give a construction of FPA with efficient encoding and decoding capabilities. Moreover, we show 
our design is   locally decodable, i.e., we can decode a message bit by reading at most $\lambda+1$ symbols, which has an
interesting application  for private information retrieval.
\end{abstract}

\section{Introduction}

Let  $n, m$ and $\lambda$ be positive integers with $n=m\lambda$, 
and $S_{n}^\lambda$ be the set of all permutations on the multiset 
$\{\overbrace{1,\dots,1}^\lambda,\dots,\overbrace{m,\dots,m}^\lambda\}$. A {\em frequency permutation array} (FPA) is a subset of $S_n^\lambda$ for some positive integers $m$, $\lambda$ and $n=m\lambda$. A $(\lambda,n,d)$-FPA is a subset of $S_{n}^\lambda$ and the distance between any pair of distinct permutations is at least $d$ under any metric, such as Hamming distance, $\ell_\infty$-norm, etc. Permutation array (PA) is  simply a special case of FPA by choosing $\lambda=1$.  
With a fixed length $n$, FPA has a smaller set of symbols than PA. Thus, codes with FPA  have  a better information rate than those with PA.
A widely adopted approach to building PAs under Hamming distance, see for example \cite{CCKT03}, is using distance preserving mappings or distance increasing mappings from $Z_2^k$ to $S_n^1$. Most of those encoding schemes are efficient but it is not clear how to decode efficiently. Lin et al. \cite{LTT08} proposed a couple of novel constructions with efficient encoding and decoding algorithms for PAs under $l_\infty$-norm. 
FPA was proposed by Huczynska and Mullen \cite{HM06} as a generalization of PA. They gave several constructions of FPA under Hamming distance and bounds for the maximum array size. In this paper, we extends the ideas in \cite{LTT08} to constructing FPA under $l_\infty$-norm. We prove lower and upper bounds of FPA.  Then we show the efficient encoding and decoding algorithms.  Besides, we show that our FPAs are locally decodable codes under $l_\infty$-norm.

Recently, researchers have found that PAs have  applications in  areas such as power line communication (e.g. \cite{Shum02}, \cite{VH00}, \cite{VHW00} and \cite{Vinck00}), multi-level flash memories (see \cite{Jiang1} and \cite{Jiang2}).   Similar to the application of PAs on power line communication, we
can encode a message as a frequency permutation from $S_n^{\lambda}$ and
associate each symbol $i \in \{1, \dots, m\}$ with a frequency $f_i$.
Then the message is transmitted as  a series of corresponding frequencies. 
For example, to send a message encoded as $(1, 2, 2, 1,3,3)$, we can 
transmit the frequency sequence $(f_1, f_2, f_2, f_1,f_3, f_3)$ one by one. 

For flash memory application, different from the approach by Jiang et. al. \cite{Jiang1,Jiang2}, we can use 
FPA to provide multi-level flash memory with error correcting capability. For example, suppose a multi-level flash memory, where each cell has $m$ states, which can be changed by injecting or removing charge into or from it.
Over injecting or charge leakage will alter the state as well.  
We can use the charge ranks of $n$ cells to represent a permutation from
$S_n^{\lambda}$, i.e., the cells with the lowest $\lambda$ charge levels represent symbol 1, and so on.
With our efficient encoding and decoding algorithms,  a $(\lambda, n, d)$-FPA can be used in flash memory system to represent information and  correct errors caused by charge level fluctuation. 
 

A locally decodable code has an extremely efficient decoding for any message bit by reading at most a fixed number of symbols from the received word. Suppose that a FPA is applied to a multi-level flash memory where the length of a codeword is nearly a block of cells (about $10^5$)\cite{FlashBook1}. This feature allows us to retrieve the desire message bits from a multi-level flash without accessing the whole block. With the locally decodable property, we can raise the robustness of the code without loss of efficiency. On the other hand, locally decodable codes have been under study for years, see \cite{Tre04} for a survey and \cite{Yek08}, \cite{Ef08} for recent progress. They are related to a cryptographic protocol called {\em private information retrieval} (PIR for short). We show our construction of FPA can also be used in cryptographic application.

{\bf Notations:}
Let $m$ and $\lambda$ be positive integers
and let $n=m \lambda$ throughout the paper unless stated otherwise.
We use $[n]$ to represent the set $\{1,\ldots,n\}$. 
$S_{n}^\lambda$ denotes the permutations over the multiset  
$\{\overbrace{1,\dots,1}^\lambda,\dots,\overbrace{m,\dots,m}^\lambda\}$.
For two vectors $\bm{x}$ and $\bm{y}$ of the same dimension, let $l_\infty(\bm{x},\bm{y})=\max_i |\bm{x}_i -\bm{y}_i|$.
We say two permutations $\bm{x}$ and $\bm{y}$ are $d$-close to each other under  metric $\delta(\cdot,\cdot)$ if $\delta(\bm{x},\bm{y})\le d$. The identity permutation $\bm{I}_n^\lambda$ in $S_{n}^\lambda$ is $(1,\dots,1,\dots,m,\dots,m)$. 

\section{Lower and upper bounds}

Let $F_\infty(\lambda,n,d)$ be the cardinality of the maximum $(\lambda,n,d)$-FPA and $V_\infty(\lambda,n,d)$ be the number of elements in $S_{n}^\lambda$ being $d$-close to the identity $\bm{I}_n^\lambda$ under $\ell_\infty$-norm. In this section, 
 we give a Gilbert type lower bound and a sphere packing upper bound of $F_\infty(\lambda,n,d)$  by bounding $V_\infty(\lambda,n,d)$. 
 
 First, we show that any $d$-radius ball in $S_{n}^\lambda$ under $l_\infty$-norm has the same cardinality.
\begin{claim}\label{ballsize}
For any $\bm{x}=(x_1,\ldots,x_n)\in S_{n}^\lambda$, there are exactly 
$V_\infty(\lambda,n,d)$ $\bm{y}$'s in $S_{n}^\lambda$ such that $l_\infty(\bm{x},\bm{y})\le d$.
\end{claim}
\begin{proof}
Since every $i\in[m]$ appears exactly $\lambda$ times in $\bm{x}$, 
there exists a permutation $\pi\in S_n^1$ such that $\bm{x}=\pi\circ\bm{I}_n^\lambda$. As a consequence, we have that $l_\infty(\bm{I}_n^\lambda,\bm{z})=l_\infty(\bm{x},\pi\circ\bm{z})$ for any $\bm{z}\in S_{n}^{\lambda}$. Let $Z=\{\bm{z}:\bm{z}\in S_{n}^{\lambda},l_\infty(\bm{I}_n^\lambda,\bm{z})\le d\}$, $Y=\{\pi\circ\bm{z}:\bm{z}\in Z\}$ and $\bar{Y}=S_{n}^{\lambda}-Y$. For any $\bm{y}\in Y$, we have $l_\infty(\bm{x},\bm{y})=l_\infty(\bm{I}_n^\lambda,\pi^{-1}\circ\bm{y})\le d$, since $\pi^{-1}\circ\bm{y}\in Z$. While  for $\bm{y}'\in \bar{Y}$, $l_\infty(\bm{x},\bm{y}')=l_\infty(\bm{I}_n^\lambda,\pi^{-1}\circ\bm{y}')> d$. Therefore, only $|Y|=|Z|=V_\infty(\lambda,n,d)$ permutations  in $S_n^{\lambda}$ are $d$-close to $\bm{x}$.
\end{proof}


\begin{theorem}\label{bounds}
\[\frac{\left|S_{n}^\lambda\right|}{V_\infty(\lambda,n,d-1)}\le F_\infty(\lambda,n,d)\le\frac{\left|S_{n}^\lambda\right|}{V_\infty(\lambda,n,\lfloor\frac{d-1}{2}\rfloor)}.\]
\end{theorem}
\begin{proof}
To prove the lower bound, we use the following algorithm to generate a $(\lambda,n,d)$-FPA with size $\ge\frac{\left|S_{n}^\lambda\right|}{V_\infty(\lambda,n,d-1)}$.
\begin{enumerate}
\item $C\leftarrow \emptyset$, $D\leftarrow S_{n}^\lambda$.
\item Add an arbitrary $\bm{x}\in D$ to $C$, then remove all permutations 
that is $(d-1)$-close to $\bm{x}$ from $D$.
\item If $D\neq\emptyset$ then repeat step 2, otherwise output $C$.
\end{enumerate}
$D$ has initially $|S_{n}^{\lambda}|$ elements and each iteration of step 2 removes at most $V_\infty(\lambda,n,d-1)$, so we conclude $|C|\ge\frac{\left|S_{n}^\lambda\right|}{V_\infty(\lambda,n,d-1)}$.

Now we turn to the upper bound. Consider a $(\lambda,n,d)$-FPA $C^*$ with the maximum cardinality. Any two $\lfloor\frac{d-1}{2}\rfloor$-radius balls centered at distinct permutations in $C^*$ do not have any common permutation, since the minimum distance is $d$. In other words, the $\lfloor\frac{d-1}{2}\rfloor$-radius balls centered at permutations in $C^*$ are all disjoint. We have $|C^*|\le\frac{|S_{n}^\lambda|}{V_\infty(\lambda,n,\lfloor\frac{d-1}{2}\rfloor)}$.
\end{proof}

It is clear that $|S_{n}^\lambda|=\frac{n!}{(\lambda!)^{n/\lambda}}$.
It is already known that $V_\infty(1,n,d)$ equals to the permanent of some special matrix \cite{LTT08}. In this paper, we generalize previous analysis to give asymptotic bounds for  Theorem \ref{bounds}. The permanent of an $n\times n$ matrix $A=(a_{i,j})$ is defined as \[{\rm per}A=\sum_{\pi\in S_n}\prod_{i=1}^na_{i,\pi_i}.\]
Define a symmetric  $n\times n$ matrix $A^{(\lambda,n,d)}=\left(a_{i,j}^{(\lambda,n,d)}\right)$, where $a_{i,j}^{(\lambda,n,d)}=1$, if $\left|\lceil\frac{i}{\lambda}\rceil-\lceil\frac{j}{\lambda}\rceil\right|\le d$; else $a_{i,j}^{(\lambda,n,d)}=0$. 
Note that a permutation $(x_1,\dots,x_n)$ is $d$-close to $\bm{I}^\lambda_n$ if and only if $a_{i,x_i}^{(\lambda,n,d)}=1$ for every $i\in[n]$.  Now we consider $A^{(\lambda,\lambda m,d)}$. 
Since the $\lambda$ copies of a symbol are considered identical while computing the distance and 
the  entries indexed from $(\ell\lambda-\lambda+1)$ to $\ell\lambda$ of $\bm{I}^\lambda_{\lambda m}$ represent
the same symbol for every $\ell\in[m]$. It implies that row $(\ell\lambda-\lambda+1)$ through row $\ell\lambda$ of $A^{(\lambda,\lambda m,d)}$ are identical and so are columns indexed from $(\ell\lambda-\lambda+1)$ to $\ell\lambda$ for every $\ell\in[m]$. Thus, we have $A^{(\lambda,\lambda m,d)}=A^{(1,m,d)}\otimes\bm{1}_\lambda$ where $\otimes$ is the operator of tensor product and $\bm{1}_\lambda$ is a $\lambda\times\lambda$ matrix with all entries equal to 1. For example, take $\lambda=2$, $m=5$ and $d=2$:
\[A^{(1,5,2)}=\left(\begin{array}{ccccc}1&1&1&0&0\\1&1&1&1&0\\1&1&1&1&1\\0&1&1&1&1\\0&0&1&1&1\end{array}\right),\bm{1}_2=\left(\begin{array}{cc}1&1\\1&1\end{array}\right)\]
\[A^{(2,10,2)}=\left(\begin{array}{cccccccccc}1&1&1&1&1&1&0&0&0&0\\1&1&1&1&1&1&0&0&0&0\\1&1&1&1&1&1&1&1&0&0\\1&1&1&1&1&1&1&1&0&0\\1&1&1&1&1&1&1&1&1&1\\1&1&1&1&1&1&1&1&1&1\\0&0&1&1&1&1&1&1&1&1\\0&0&1&1&1&1&1&1&1&1\\0&0&0&0&1&1&1&1&1&1\\0&0&0&0&1&1&1&1&1&1\end{array}\right)\]
Let $r^{(1,m,d)}_i$ be the row sum of  $A^{(1,m,d)}$'s $i$-th row. We have:
\[r^{(1,m,d)}_i=\left\{\begin{array}{ll}d+i&\mbox{if } i\le d,\\2d+1&\mbox{if } d<i\le m-d,\\m-i+1+d&\mbox{if } i>m-d.\end{array}\right.\] 
Then for $i\in[m]$ and $j\in[\lambda]$, the row sum of the $(i\lambda-\lambda+j)$-th row of $A^{(\lambda,\lambda m,d)}$ is $\lambda r^{(1,m,d)}_i$, due to $A^{(\lambda,\lambda m,d)}=A^{(1,m,d)}\otimes\bm{1}_\lambda$.
We first calculate  $V_\infty(\lambda,n,d)$ by using ${\rm per}A^{(\lambda,n,d)}$.

\begin{lemma}
\[V_\infty(\lambda,n,d)={{{\rm per}A^{(\lambda,n,d)}}\over{(\lambda!)^{n/\lambda}}}.\]
\end{lemma}
\begin{proof}
\[\begin{array}{ll}
&{\rm per}A^{(\lambda,n,d)}\\
=&|\{\bm{x}\in S_{n}^1:\forall i, a^{(\lambda,n,d)}_{i,x_i}=1\}|\\
=&|\{\bm{x}\in S_{n}^1:\max_i|\lceil\frac{i}{\lambda}\rceil-\lceil\frac{x_i}{\lambda}\rceil|\le d\}|\\
=&(\lambda!)^{n/\lambda}|\{\bm{y}\in S_{n}^\lambda:\max_i|\lceil\frac{i}{\lambda}\rceil-y_i|\le d\}|\\
=&(\lambda!)^{n/\lambda}|\{\bm{y}\in S_{n}^\lambda:l_\infty(\bm{I}_n^\lambda,\bm{y})\le d\}|\\
=&(\lambda!)^{n/\lambda}V_{\infty}(\lambda,n,d)
\end{array}\]
The first equality holds since $A^{(\lambda,n,d)}$ is a $(0,1)$-matrix and by the definition of permanent. 
We can convert $\bm{x}\in S_n^1$ into $\bm{y}\in S_{n}^\lambda$ by setting $y_i=\lceil\frac{x_i}{\lambda}\rceil$, and there are exactly $(\lambda!)^{n/\lambda}$ $\bm{x}$'s in $S_n^1$ converted to the same $\bm{y}$. Thus, we know the third equality holds.
Therefore, the lemma holds by moving $(\lambda!)^{n/\lambda}$ to the left-hand side of the equation.
\end{proof}

We still need to estimate ${\rm per}A^{(\lambda,n,d)}$  
in order to get asymptotic bounds.  
Kl{\o}ve \cite{Klove08} reports some bounds and methods to approximate ${\rm per}A^{(1,n,d)}$. We extend his analysis for ${\rm per}A^{(\lambda,n,d)}$.

\begin{lemma}\label{permU}
${\rm per}A^{(\lambda,n,d)}\le \left[(2d\lambda+\lambda)!\right]^{\left.\frac{n}{2d\lambda+\lambda}\right.}.$
\end{lemma}

\begin{proof}
It is known (Theorem 11.5 in \cite{LW03}) that for $(0,1)$-matrix $A$, ${\rm per}A\le\prod_{i=1}^n(r_i!)^{\frac{1}{r_i}}$ where $r_i$ is the sum of the $i$-th row. Since the sum of any row of $A^{(\lambda,n,d)}$ is at most $2d\lambda+\lambda$, we have ${\rm per}A\le\prod_{i=1}^n[(2d\lambda+\lambda)!]^{\frac{1}{2d\lambda+\lambda}}=[(2d\lambda+\lambda)!]^{\frac{n}{2d\lambda+\lambda}}$.
\end{proof}
We  give ${\rm per}A^{(\lambda,n,d)}$ a lower bound by using the van der Waerden permanent theorem (see p.104 in \cite{LW03}): {\em the permanent of an $n\times n$ doubly stochastic matrix $A$ (i.e., $A$ has nonnegative entries, and every row sum and column sum of $A$ is 1.) is no less than $\frac{n!}{n^n}$.} Unfortunately, $A^{(\lambda,n,d)}$ is not a doubly stochastic matrix, since the row sums and columns sums range from $d\lambda+\lambda$ to $2d\lambda+\lambda$.
We estimate the lower bound via a matrix derived from $A^{(\lambda,n,d)}$ as follows.

\begin{lemma}\label{permL}
${\rm per}A^{(\lambda,n,d)}\ge \frac{(2d\lambda+\lambda)^n}{2^{2d\lambda}}\cdot\frac{n!}{n^n}$.
\end{lemma}

\begin{proof}
Let $\tilde{A}=\frac{1}{2d\lambda+\lambda}A^{(\lambda,n,d)}$, which has the sum of any row or column bounded by $1$, but is not a doubly stochastic matrix. Observe that every row sum of $\tilde{A}$ is $1$ except the first $d\lambda$ and last $d\lambda$ rows. For $i\in[d]$ and $j\in[\lambda]$, both row $(i\lambda-\lambda+j)$ and row $(n-i\lambda+j)$ sum to $\frac{d+i}{2d+1}$. Now we construct an $n\times n$ matrix $B$ from $\tilde{A}$ with each row sum equal to $1$ as follows: 
\begin{tabbing}
For $i\in[d]$ and $j\in[\lambda]$, add $\frac{1}{2d\lambda+\lambda}$ to\\
1) The first $(d-i+1)\lambda$ entries of row $(i\lambda-\lambda+j)$.\\
2) The last $(d-i+1)\lambda$ entries of row $(n-i\lambda+j)$.
\end{tabbing}
The row sums of the first $d\lambda$ and last $d\lambda$ rows of $B$ are now $\frac{(d-i+1)\lambda}{2d\lambda+\lambda}+\frac{d+i}{2d+1}=1$.

We turn to check the column sums of $B$. Since $\tilde{A}$ is symmetric and by the definition of $B$, we know $B$ is symmetric as well.
Thus we have that $B$ is doubly stochastic and ${\rm per}B\ge\frac{n!}{n^n}$.

Now we turn to bound ${\rm per}A^{(\lambda,n,d)}$. Observe that the entries of the first $d\lambda$ and last $d\lambda$ rows of $B$ are at most $\frac{2}{2d\lambda+\lambda}$ times of the corresponding entries of $A^{(\lambda,n,d)}$, and the other rows are exactly $\frac{1}{2d\lambda+\lambda}$ times of the corresponding rows of $A^{(\lambda,n,d)}$. We have ${\rm per}A^{(\lambda,n,d)}\ge \frac{(2d\lambda+\lambda)^n}{2^{2d\lambda}}{\rm per}B\ge\frac{(2d\lambda+\lambda)^n}{2^{2d\lambda}}\frac{n!}{n^n}$. 
\end{proof}

With Lemma \ref{permU} and Lemma \ref{permL}, we have the asymptotic bounds as follows.

\begin{theorem}
\[\frac{n!}{\left[(2d\lambda-\lambda)!\right]^{\left.\frac{n}{2d\lambda-\lambda}\right.}}\le F_\infty(\lambda,n,d)\le\frac{2^{2\lambda\cdot\lfloor\frac{d-1}{2}\rfloor}n^n}{(2\lambda\cdot\lfloor\frac{d-1}{2}\rfloor+\lambda)^n}.\]
\end{theorem}

\section{Encoding and decoding}

Our construction idea is based on the previous work\cite{LTT08} by Lin, et al. We generalize their algorithm for constructing FPAs. Furthermore, we give the first locally decoding algorithm for FPAs under $l_\infty$-norm. 

\subsection{Encoding algorithm}
We give an encoding algorithm $E_{n,k}^\lambda$ (see Figure \ref{encAlgo}) which convert $k$-bit message into a permutation in $S_{n}^\lambda$ where $n \geq k+\lambda$. 

\begin{center}
\begin{figure}[htbp]\begin{tabbing}
{\bf Algorithm $E_{n,k}^\lambda$} \\
{\bf Input:}  $(m_1,\dots,m_{k}) \in Z_2^{k}$\\
{\bf Output:} $(x_1,\dots,x_{n}) \in S_{n}^\lambda$\\
\hspace{.2cm}  $max \leftarrow n;$ $min \leftarrow 1$;\\
\hspace{.2cm} {\bf for} $i\leftarrow 1$ {\bf to} $k$ {\bf do}\\
\hspace*{.6cm}\,\, {\bf if} \= $m_{i}=1$\\
\>{\bf then }\=\{$x_i \leftarrow \lceil\frac{max}{\lambda}\rceil$; $max \leftarrow max - 1$;\}\\
\> {\bf else}\>\{$x_i \leftarrow \lceil\frac{min}{\lambda}\rceil$; $min \leftarrow min + 1$;\}\\
\hspace{.2cm} {\bf for} $i\leftarrow k+1$ {\bf to} $n$ {\bf do}\\
\hspace*{.6cm}\,\, $x_i \leftarrow \lceil\frac{min}{\lambda}\rceil$; $min \leftarrow min + 1$;\\
Output $(x_1,\dots,x_{n})$.
\end{tabbing}
\caption{$E_{n,k}^\lambda$ encodes messages in $Z_2^k$ with $S_{n}^\lambda$.} \label{encAlgo}
\end{figure}
\end{center}
The encoding algorithm $E_{n,k}^\lambda$ maps binary vectors from 
$Z_2^k$ to  $S_{n}^\lambda$ and it is a distance preserving mapping.
It is clear that $E_{n,k}^\lambda$ runs in $O(n)$ time while encoding any $k$-bit message. Next we investigate the properties of the code obtained by $E_{n,k}^\lambda$. Let $C_{n,k}^\lambda$ be the image of $E_{n,k}^\lambda$. 
\begin{theorem}\label{encTHM}
$C_{n,k}^\lambda$ is a $(\lambda,n,\lfloor\frac{n-k}{\lambda}\rfloor)$-FPA with cardinality $2^k$.
\end{theorem}
\begin{proof}
Consider two messages $\bm{p}=$ $(p_1,\dots,p_k)$ and $\bm{q}=(q_1,\dots,q_k) \in Z_2^k$. Let $\bm{x^p}$ and $\bm{x^q}$ be the outputs of $E_{n,k}^\lambda$, respectively. Let $r$ be the smallest index such that $p_r\neq q_r$. Without loss of generality, we assume $p_r=1$, $q_r=0$ and there are exactly $z$ zeroes among $p_1,\dots,p_{r-1}$. Consequently, $x_r^{\bm{p}}$ is set to $\lceil\frac{max}{\lambda}\rceil=\lceil\frac{n-r+1+z}{\lambda}\rceil$ and $x_r^{\bm{q}}$ is set to $\lceil\frac{min}{\lambda}\rceil=\lceil\frac{1+z}{\lambda}\rceil$ by $E_{n,k}^\lambda$ . 
The distance between  $\bm{x^p}$ and $\bm{x^q}$ is:

\begin{eqnarray*}
&&\left\lceil\frac{n-r+1+z}{\lambda}\right\rceil-\left\lceil\frac{1+z}{\lambda}\right\rceil \\
 > &&\frac{n-r+1+z}{\lambda} - \frac{1+z}{\lambda}-1\\
 = &&\frac{n-r}{\lambda} -1\\
\ge &&\frac{n-k}{\lambda} -1, \mbox{  since $r\le k$}.
\end{eqnarray*} 

The first inequality holds by the fact of ceiling function: 
$ a \le \left\lceil a \right\rceil < a +1,$ for any real number $a$.
Note that the distance has integer value only here.  If $\frac{n-k}{\lambda}$
is integer then the distance is  at least $\left\lfloor\frac{n-k}{\lambda}\right\rfloor$; else it is at least $\left\lceil\frac{n-k}{\lambda}-1\right\rceil$, which is 
$\left\lfloor\frac{n-k}{\lambda}\right\rfloor$ exactly, i.e., the distance between any two codewords in $C_{n,k}^\lambda$ is  at least $\lfloor\frac{n-k}{\lambda}\rfloor$. 
Since every  message is encoded into a distinct codeword, we have 
$C_{n,k}^\lambda=2^k$. \end{proof}

Since $C_{n,k}^\lambda$ is a $(\lambda,n,\lfloor\frac{n-k}{\lambda}\rfloor)$-FPA, we let $d=\lfloor\frac{n-k}{\lambda}\rfloor$  for convenience. 

\subsection{Unique decoding algorithm}

Unique decoding algorithms for classic error correcting codes are usually much more complicated than their encoding algorithms. While, 
our proposed decoding algorithm $U_{n,k}^\lambda$ (see Figure \ref{uniDec}) remains simple.  

\begin{center}
\begin{figure}[htbp]\begin{tabbing}
{\bf Algorithm $U_{n,k}^\lambda$} \\
{\bf Input:} $(x_1,\dots,x_{n}) \in S_{n}^\lambda$\\
{\bf Output:}  $(m_1,\dots,m_{k}) \in Z_2^{k}$\\
\hspace{.2cm}  $max \leftarrow n;$ $min \leftarrow 1$;\\
\hspace{.2cm} {\bf for} $i\leftarrow 1$ {\bf to} $k$ {\bf do}\\
\hspace*{.6cm}\,\, {\bf if} \= $|x_i-\lceil\frac{max}{\lambda}\rceil|<|x_i-\lceil\frac{min}{\lambda}\rceil|$\\
\>{\bf then }\=\{$m_i \leftarrow 1$; $max \leftarrow max - 1$;\}\\
\>{\bf else}\>\{$m_i \leftarrow 0$; $min \leftarrow min + 1$;\}\\
Output $(m_1,\dots,m_{k})$.
\end{tabbing}

\caption{$U_{n,k}^\lambda$ decodes words in $S_{n}^\lambda$ to messages in $Z_2^k$.} \label{uniDec}
\end{figure}
\end{center}

The running time of $U_{n,k}^\lambda$ is clearly $O(k)$, even faster than the encoding algorithm. We show its correctness as follows.

\begin{theorem}
Given a permutation $\bm{x}=(x_1,\dots,x_n)$ which is $\frac{d-1}{2}$-close to  $E_{n,k}^\lambda(\bm{m})$ for some $\bm{m}\in Z_{2}^k$, algorithm $U_{n,k}^\lambda$ outputs $\bm{m}$ correctly.
\end{theorem}
\begin{proof}
By contradiction, assume $U_{n,k}^\lambda$ outputs $\hat{\bm{m}}=(\hat{m}_1,\cdots,\hat{m}_k)\neq \bm{m}$. Let $E_{n,k}^\lambda(\bm{m})=(y_1,\dots,y_n)$, $r$ be the smallest index such that $m_r\neq \hat{m}_r$ and $z$ be the number of zeroes among $m_1,\dots,m_{r-1}$. 
At the beginning of the $r$-th iteration, $max=n-r+1+z$ and $min=1+z$ because for every $i<r$, $m_i=\hat{m}_i$. Without loss of generality,  assume $1=m_r\neq\hat{m}_r=0$.
Note that $y_r$ is set to $\lceil\frac{max}{\lambda}\rceil$=$\lceil\frac{n-r+1+z}{\lambda}\rceil$ by $E_{n,k}^\lambda$.
While $\hat{m}_r$ is decoded to 0 by $U_{n,k}^\lambda$, we have 
$|x_r-\lceil\frac{max}{\lambda}\rceil| \ge |x_r-\lceil\frac{min}{\lambda}\rceil|$.
Thus,
\[\begin{array}{rcl}l_\infty(\bm{x},E_{n,k}^\lambda(\bm{m}))&\ge&|x_r-y_r|=|x_r-\lceil\frac{max}{\lambda}\rceil|\\&\ge&\frac{1}{2}\left(|x_r-\lceil\frac{max}{\lambda}\rceil|+|x_r-\lceil\frac{min}{\lambda}\rceil|\right)\\
&\ge&\frac{1}{2}\left(\lceil\frac{max}{\lambda}\rceil-\lceil\frac{min}{\lambda}\rceil\right)\\&=&\frac{1}{2}\left(\lceil\frac{n-r+1+z}{\lambda}\rceil-\lceil\frac{1+z}{\lambda}\rceil\right)\ge\frac{d}{2}.
\end{array}\]
The last inequality is true, since we know $\lceil\frac{n-r+1+z}{\lambda}\rceil-\lceil\frac{1+z}{\lambda}\rceil\ge\left\lfloor\frac{n-k}{\lambda}\right\rfloor=d$ from the proof of Theorem \ref{encTHM}.  This contradicts that 
$\bm{x}$  is $\frac{d-1}{2}$-close to  $E_{n,k}^\lambda(\bm{m})$.
\end{proof}

\subsection{Locally decoding algorithm}

Next we show a locally decoding algorithm $L_{n,k}^\lambda$, see Figure \ref{locDec}, which is a probabilistic algorithm. We discuss its efficiency and error probability in this subsection. We prove that it reads at most $\lambda+1$ entries of the received word in Lemma \ref{term}, hence its running time is $O(\lambda)$. It has a chance to output wrongly, but we show that the error probability is small in Theorem \ref{err}. Furthermore, $L_{n,k}^\lambda$ always outputs correct message bit when it was given a codeword as input, see Corollary \ref{code}.

\begin{center}
\begin{figure}[htbp]\begin{tabbing}
{\bf Algorithm $L_{n,k}^\lambda$} \\
{\bf Input:} $i\in[n],(x_1,\dots,x_{n}) \in S_{n}^\lambda$\\
{\bf Output:}  $m_i$, the $i$-th message bit\\
\hspace{.2cm}  $J\leftarrow\{i+1,\dots,n\}$;\\
\hspace{.2cm} {\bf do}\\
\hspace{.6cm}  Uniformly and randomly pick $j\in J$;\\
\hspace{.6cm} {\bf if} \= $x_{i}>x_j$~{\bf then} output 1;\\
\hspace{.6cm} {\bf if} \= $x_{i}<x_j$~{\bf then} output 0;\\
\hspace{.6cm}  $J\leftarrow J-\{j\}$;\\
\hspace{.2cm} {\bf loop};
\end{tabbing}
\caption{$L_{n,k}^\lambda$ decodes one bit by reading at most $\lambda+1$ symbols.} \label{locDec}
\end{figure}
\end{center}

\begin{lemma}\label{term}
Given a permutation $\bm{x}=(x_1,\dots,x_n)\in S_{n,k}^\lambda$ and an index $i\in[k]$, $L_{n,k}^\lambda$ terminates within $\lambda$ iterations.
\end{lemma}
\begin{proof}
By contradiction, assume $L_{n,k}^\lambda$ does not output before the end of the $\lambda$-th iteration. For $\ell\le\lambda$, let $j_\ell$ be the index picked in the $\ell$-th iteration. For every $\ell\le\lambda$, we have $x_i=x_{j_\ell}$, otherwise $L_{n,k}^\lambda$ outputs at the $\ell$-th iteration. Therefore, there are at least $\lambda+1$ entries of $\bm{x}$ equal to $x_i$. It implies $\bm{x}\notin S_{n,k}^\lambda$, a contradiction. There is some $x_{j_\ell}\neq x_i$, and $L_{n,k}^\lambda$ outputs in the $\ell$-th iteration.
\end{proof}

\begin{theorem}\label{err}
Given a permutation $\bm{x}=(x_1,\dots,x_n)$ $\delta$-close to a codeword $E_{n,k}^\lambda(\bm{m})=(y_1,\dots,y_n)\in S_{n,k}^\lambda$ for some $\bm{m}$ and an index $i\in[k]$, $L_{n,k}^\lambda$ outputs $m_i$ with probability at least $1-\frac{2\delta+1}{d}$ at its first iteration.
\end{theorem}

\begin{proof}
Without loss of generality, we assume $m_i=0$, $y_i=t$ and  let $u$ be the maximum number among $y_{i+1},\dots,y_{n}$, i.e., at the start of the $i$-th iteration $\min=t$ and $\max=u$ while encoding. Assume there are $\gamma$ numbers equal to $t$ among $y_1,\dots,y_{i-1}$, and there are  $\gamma'$ numbers equal to $u$ among $y_{i+1},\dots,y_{n}$. According to the encoding algorithm, we have \[\{y_{i+1},\ldots,y_{n}\}=\{\overbrace{t,\dots,t}^{\lambda-\gamma-1},\overbrace{t+1,\dots,t+1}^\lambda,\dots,\overbrace{u,\dots,u}^{\gamma'}\}\] Since $l_\infty(\bm{x},E_{n,k}^\lambda(\bm{m}))\le\delta$, we have $|x_j-y_j|\le \delta$ and $|x_i-y_i|\le \delta$. The probability that $L_{n,k}^\lambda$ does not output $m_i$ at the first iteration is:
\[\begin{array}{rcl}
\Pr[x_i\ge x_j]&\le&\Pr[y_i+\delta \ge x_j]\\
                    &\le&\Pr[y_i+\delta\ge y_j-\delta]\\
                    &=&\Pr[y_i+2\delta\ge y_j].\end{array}\]
There are at most $2\delta\lambda+\lambda-\gamma-1$ possible $y_j$'s less than or equal to $y_i+2\delta$.  Thus,  \[\Pr[x_i\ge x_j]\le\frac{(2\delta+1)\lambda-\gamma-1}{n-i}\le\frac{2\delta\lambda+\lambda}{d\lambda}=\frac{2\delta+1}{d}.\]
Therefore, the probability that $L_{n,k}^\lambda$ outputs $m_i$ correctly at the first iteration is at least $1-\frac{2\delta+1}{d}$.
\end{proof}

\begin{corollary}\label{code}
Given a codeword $\bm{x}=E_{n,k}^\lambda(\bm{m})$ for some $\bm{m}$ and an index $i$, $L_{n,k}^\lambda$ 
outputs $m_i$ correctly.
\end{corollary}

\begin{proof}
By Lemma \ref{term},  there exists $\ell\le \lambda$ such that $L_{n,k}^\lambda$ terminates at the $\ell$-th iteration. Let $j$ be the index picked at the $\ell$-th iteration, we have $x_j\neq x_i$, where $j>i$. 
Note that $\bm{x}$ is a codeword: $x_i< x_j$ implies $m_i=0$ and $x_i>x_j$  implies $m_i=1$. Hence, $L_{n,k}^\lambda$ outputs $m_i$ correctly.
\end{proof}

A private information retrieval system (PIR) consists of $q$ servers. All servers know a codeword $\bm{x}=(x_1,\dots,x_n)$ representing a message $\bm{m}=(m_1,\dots,m_k)$, and a user wants to know one bit $m_i$ of $\bm{m}$ via query a symbol from each server. We say a PIR has  {\em retrievability} $r$ if the user can obtain the message bit with probability $r$. Let $\mathcal{D}(s,i)$ be the distribution of entry queried from server $i$ when the user tries to retrieve $m_i$. A PIR has {\em privacy} $p$ if $\max_{i,j\in[k],s\in[q]}\Delta(\mathcal{D}(s,i),\mathcal{D}(s,j))\le p$, where $\Delta(\cdot,\cdot)$ is the statistical distance. A $(q,r,p)$-PIR is a $q$-server PIR with retrievability $r$ and privacy $p$. A $(q,r,p)$-PIR has perfect retrievability if $r=1$ and perfect privacy if $p=0$.

With our FPA $C_{n,k}^\lambda$, we construct a $(\lambda+1,1,r)$-PIR with perfect retrievability and privacy $r$. 
The scheme is simple:
\begin{itemize}
\item For a message $\bm{m}$, we put $\bm{x}=E_{n,k}^\lambda(\bm{m})$ on all $\lambda+1$ servers.
\item We retrieve $m_i$ by $L_{n,k}^\lambda$ by querying entries from servers in a random order.
\end{itemize}
The perfect retrievability is guaranteed by Corollary 1. However, in order to retrieve $m_i$, $x_i$ must be queried from some servers at certain positions $\ell > i$, and we have $r > 0$.  We leave the improvement on the privacy $r$ as our future work.

\end{document}